\documentclass[pdflatex,sn-mathphys-num]{sn-jnl}
\usepackage{graphicx}
\usepackage{placeins}

\usepackage{amsmath,amssymb,amsfonts,amsthm}
\usepackage{mathtools,braket,bm,booktabs}
\usepackage[title]{appendix}
\theoremstyle{thmstyleone}
\newtheorem{theorem}{Theorem}
\newtheorem{corollary}[theorem]{Corollary}
\newtheorem{proposition}[theorem]{Proposition}
\DeclareMathOperator{\tr}{tr}
\DeclareMathOperator{\supp}{supp}
\DeclareMathOperator{\Ad}{Ad}
\newcommand{\Id}{\mathrm{id}}
\newcommand{\FO}{\mathrm{FO}}
\newcommand{\SW}{\mathrm{SW}}
\newcommand{\HS}{\mathrm{HS}}
\newcommand{\cE}{\mathcal E}
\newcommand{\cS}{\mathcal S}
\newcommand{\cH}{\mathcal H}
\begin{document}

\title[Commutator geometry and information preservation]{Commutator Geometry and Information Preservation in the Quantum Switch}
\author*[1]{\fnm{Xu} \sur{Chen}}\email{cnc@hebust.edu.cn}
\author[2]{\fnm{Xue} \sur{Ma}}
\affil*[1]{\orgdiv{School of Sciences}, \orgname{Hebei University of Science and Technology}, \orgaddress{\city{Shijiazhuang}, \postcode{050018}, \state{Hebei}, \country{People's Republic of China}}}
\affil[2]{\orgdiv{Network Management Center}, \orgname{China Mobile Communications Group Hebei Co., Ltd.}, \orgaddress{\city{Shijiazhuang}, \postcode{050000}, \state{Hebei}, \country{People's Republic of China}}}

\abstract{Two commuting channels yield the same composite channel in either fixed order, yet their quantum switch can alter information preservation. We study how this effect depends on the input state, retaining the joint output of the control and target. Exact Kraus commutator identities determine overlap and squared fidelity gaps. Fixed order preserves product inputs at least as well as the switch in fidelity. The reverse inequality holds for pure states of two qubits whose control Schmidt basis can be chosen to coincide with the order basis. A commutator matrix determines the target basis that maximizes the squared fidelity gap at fixed entanglement within this aligned family. For Pauli $X$ and $Y$ channels with equal error probabilities, the switch replaces a logical phase flip with an operator that stabilizes a joint code. At fixed noise, the switch preserves at least as much distinguishability and quantum Fisher information as fixed order for every state family in this code. For the pure encoded family studied here, exact formulas connect the phase information gain to the squared fidelity gap and quantify the reduction in information about the noise probability. Optimizing the probes yields equal maximal quantum Fisher information for noise estimation under both joint channels. These results connect commutator geometry and logical error structure to state preservation and the retention of encoded information.}

\keywords{quantum switch, quantum fidelity, Kraus commutators, quantum Fisher information}
\maketitle

\section{Introduction}\label{sec:introduction}

The quantum switch coherently controls the order of two operations,
connecting quantum causal structure with the processing of quantum
information~\cite{Hardy2007,Oreshkov2012,Brukner2014,Chiribella2013}.
Interference between orders enables advantages in channel
discrimination~\cite{Chiribella2012}, quantum
computation~\cite{Araujo2014}, and communication
complexity~\cite{Guerin2016}. Experiments have demonstrated superpositions
of gate orders~\cite{Procopio2015} and verified causal nonseparability
in the quantum switch~\cite{Rubino2017,Goswami2018}.

For noisy channels, coherent control of order can change the information
transmitted through a sequence of operations. Two identical completely
depolarizing channels can transmit classical information when combined
in a quantum switch~\cite{Ebler2018}. Quantum transmission through noisy
channels has also been studied theoretically~\cite{Salek2018,Caleffi2020}
and demonstrated experimentally~\cite{Guo2020}, while protocols involving
teleportation and Pauli channels show improvements in conditional
fidelity~\cite{Mukhopadhyay2020,Delgado2020}. Identical depolarizing
channels illustrate the distinction between channel commutativity and
Kraus commutativity: the channels commute under composition, although
their Kraus operators need not commute. The two fixed orders can therefore
yield the same composite channel while the quantum switch changes the
joint output.

These results raise the question of how interference between orders
affects the preservation of quantum information. We address this
question by comparing the quantum switch with fixed order for two
commuting channels. Both processes use each channel once and retain
the joint output of the control and target. We examine preservation
through fidelity with the input state and the distinguishability of
encoded state families.

These measures describe different aspects of information preservation.
Fidelity measures how well a given input state is preserved, whereas
distinguishability measures how well the output retains differences
between inputs. Quantum Fisher information (QFI) quantifies local
distinguishability within a parameterized state family and connects it
to parameter estimation~\cite{Braunstein1994}. Metrological applications
of indefinite causal order include displacement
estimation~\cite{Zhao2020}, estimation of noisy unitary
parameters~\cite{Chapeau2021,An2024}, and joint estimation of signal
and noise parameters~\cite{Goldberg2023}. General optimization
frameworks determine the attainable precision for different classes
of strategies~\cite{Liu2024}. Here, we determine how input geometry
controls the fidelity comparison and how the action of errors on a
code relates this comparison to the preservation of parameter information.

Measurements of the switched control can quantify Kraus
noncommutativity~\cite{Gao2023}. Here, exact Kraus commutator identities
isolate the contribution selected by the input and determine the
overlap and squared fidelity gaps between the two processes.
For every product input, the fidelity under fixed order is at least
as high as under the switch. The reverse inequality holds for pure
states of two qubits whose control Schmidt basis can be chosen to
coincide with the order basis. We call these states aligned inputs.

The same identities describe the input geometry. Within the aligned
family, a commutator matrix determines the target basis that maximizes
the squared fidelity gap at fixed entanglement. Control rotations show
that inputs with the same entanglement can have gaps of opposite sign,
while qutrit examples establish the role of target dimension. These
results identify input orientation and accessible dimension as essential
to the fidelity comparison, complementing studies of entanglement
generation~\cite{Aslanbas2026} and certification~\cite{Wang2026}
through coherent control of order.

For Pauli $X$ and $Y$ channels with equal error probabilities,
the switch replaces a logical phase flip with an operator that
stabilizes a joint code. An explicit channel maps the switched output
of every encoded state to its fixed order counterpart. At fixed noise,
this relation establishes that the switch preserves at least as much
distinguishability and QFI as fixed order for every state family in
the code. For the pure encoded family studied here, exact formulas connect the phase
QFI gain to the squared fidelity gap and quantify the reduction in
information about the noise probability for the same probe. Optimizing
the probes gives equal maximal noise QFI for both joint channels.
These results connect commutator geometry and logical error structure
to state preservation and the retention of encoded information.

Section~\ref{sec:model} defines the joint channels and comparison
measures. Sections~\ref{sec:gaps}--\ref{sec:dimension} establish the
commutator gaps and their dependence on input geometry and target
dimension. Section~\ref{sec:encoding} develops the Pauli code and its
information properties. Section~\ref{sec:conclusion} summarizes the
results and their implications. Additional derivations are given in
the appendices.

\section{Joint fidelity and interference between orders}\label{sec:model}

Let the control space be $\cH_C\cong\mathbb C^2$, with a fixed order basis $\{\ket0_C,\ket1_C\}$, and let $\cH_T\cong\mathbb C^d$. Using the operator sum description of quantum operations~\cite{Kraus1971}, equivalent to a dilation by an environment~\cite{Stinespring1955}, consider two completely positive trace preserving (CPTP) channels~\cite{NielsenChuang2010}
\begin{equation}
 \Phi(\rho)=\sum_i K_i\rho K_i^\dagger,\qquad
 \Psi(\rho)=\sum_j L_j\rho L_j^\dagger,
 \qquad \Phi\circ\Psi=\Psi\circ\Phi=\cE.
 \label{eq:channels}
\end{equation}
The quantum switch and the fixed order map are
\begin{align}
 \cS(\rho)&=\sum_{ij}W_{ij}\rho W_{ij}^\dagger,
 &W_{ij}&=\ket0\!\bra0_C\otimes K_iL_j+
 \ket1\!\bra1_C\otimes L_jK_i,\label{eq:switch}\\
 \mathcal F(\rho)&=(\Id_C\otimes\cE)(\rho).
 &&\label{eq:fixed}
\end{align}
Each protocol uses each channel once. We evaluate preservation under the joint maps in Eqs.~\eqref{eq:switch} and~\eqref{eq:fixed}, retaining the complete joint output.

For an input $\rho$, write $\rho_{\SW}=\cS(\rho)$ and $\rho_{\FO}=\mathcal F(\rho)$. We use the root fidelity~\cite{Uhlmann1976,Jozsa1994,Watrous2018}
\begin{equation}
 F(\rho,\sigma)=\tr\sqrt{\sqrt\rho\,\sigma\sqrt\rho},\qquad
 F_X=F(\rho,\rho_X),\qquad X\in\{\FO,\SW\}.
 \label{eq:fidelity}
\end{equation}
Fidelity with the input quantifies state preservation~\cite{Gilchrist2005}. We also write $P_X=\tr(\rho\rho_X)$ and $\Delta P=P_{\SW}-P_{\FO}$. For pure inputs, $P_X=F_X^2$, so $\Delta P$ and the fidelity gap have the same sign. For mixed inputs, $P_X$ denotes the Hilbert--Schmidt overlap and $F_X$ denotes root fidelity.

To see where the two outputs differ, expand the input as
$\rho=\sum_{a,b=0}^1\ket a\!\bra b_C\otimes\rho_{ab}$ and define
\begin{equation}
 \mathcal M(A)=\sum_{ij}K_iL_j A(L_jK_i)^\dagger.
 \label{eq:crossmap}
\end{equation}
The switch applies $\cE$ to both diagonal control blocks, but applies $\mathcal M$ to the $01$ block. The fixed order applies $\cE$ to every block. In particular,
\begin{equation}
 \tr_C\rho_{\SW}=\cE(\rho_{00}+\rho_{11})
 =\tr_C\rho_{\FO}.
 \label{eq:marginal}
\end{equation}
The fixed order preserves the reduced control state exactly. The switch encodes the order sensitive interference in the off-diagonal joint blocks through $\mathcal M$.

Define the symmetric and antisymmetric amplitudes
\begin{equation}
 S_{ij}=\frac{K_iL_j+L_jK_i}{2},\qquad
 H_{ij}=\frac{K_iL_j-L_jK_i}{2},\qquad
 Q_{ij}=[K_i,L_j]=2H_{ij}.
 \label{eq:symmetric}
\end{equation}
Then
\begin{equation}
 W_{ij}=I_C\otimes S_{ij}+Z_C\otimes H_{ij}.
 \label{eq:order-amplitude}
\end{equation}
Thus the amplitude that changes sign upon reversing the order is coupled to $Z_C$. Channel commutativity gives the operator identity
\begin{equation}
 \mathcal D(A):=\sum_{ij}Q_{ij}A Q_{ij}^\dagger
 =2\cE(A)-\mathcal M(A)-\mathcal M(A)^\dagger
 \quad (A=A^\dagger).
 \label{eq:defect}
\end{equation}
For $A\geq0$, $\mathcal D(A)\geq0$. This positivity determines the product bound, while the coupling to $Z_C$ determines how the input correlations change its sign.

\section{Exact commutator gaps}\label{sec:gaps}

\subsection{Product inputs}

For every product input $\rho=\tau_C\otimes\rho_T$, fidelity monotonicity under partial trace and Eq.~\eqref{eq:marginal} give
\begin{equation}
 F_{\SW}\leq F(\rho_T,\cE(\rho_T))=F_{\FO}.
 \label{eq:all-product-bound}
\end{equation}
A pure control resolves this data processing bound~\cite{Watrous2018} into an exact commutator sum.

\begin{theorem}[Product loss in arbitrary target dimension]\label{thm:product}
Let $\Phi$ and $\Psi$ be commuting CPTP maps on a finite-dimensional target, with joint maps defined by Eqs.~\eqref{eq:switch} and~\eqref{eq:fixed}. Let $\rho=\ket\phi\!\bra\phi_C\otimes\rho_T$, with
$\ket\phi=a\ket0+b\ket1$, $|a|^2+|b|^2=1$, and $q=|a|^2|b|^2$. Then
\begin{equation}
 P_{\FO}-P_{\SW}
 =q\sum_{ij}\left\|\sqrt{\rho_T}\,Q_{ij}\sqrt{\rho_T}\right\|_{\HS}^2.
 \label{eq:product-gap}
\end{equation}
Both $P_{\SW}\leq P_{\FO}$ and $F_{\SW}\leq F_{\FO}$ hold. For $q>0$, equality in either inequality is equivalent to
\begin{equation}
 \Pi_T Q_{ij}\Pi_T=0\quad\text{for every }i,j,
 \label{eq:product-equality}
\end{equation}
where $\Pi_T$ projects onto $\supp\rho_T$. If $q=0$, the two outputs coincide. For a pure target $\rho_T=\ket\psi\!\bra\psi$, Eq.~\eqref{eq:product-gap} becomes
\begin{equation}
 \ F_{\FO}^2-F_{\SW}^2
 =q\sum_{ij}|\bra\psi Q_{ij}\ket\psi|^2.\ 
 \label{eq:pure-product-gap}
\end{equation}
\end{theorem}

\begin{proof}
Projecting onto the input control and using Eq.~\eqref{eq:defect} gives
\begin{equation}
 \bra\phi\rho_{\SW}\ket\phi_C
 =\cE(\rho_T)-q\mathcal D(\rho_T).
 \label{eq:projected-switch}
\end{equation}
Tracing against $\rho_T$ proves Eq.~\eqref{eq:product-gap}. Define
\begin{equation}
 A=\sqrt{\rho_T}\,\cE(\rho_T)\sqrt{\rho_T},\qquad
 B=q\sqrt{\rho_T}\,\mathcal D(\rho_T)\sqrt{\rho_T}.
\end{equation}
Then $F_{\FO}=\tr\sqrt A$ and $F_{\SW}=\tr\sqrt{A-B}$. Since $0\leq A-B\leq A$, operator monotonicity gives $F_{\SW}\leq F_{\FO}$. Equality holds exactly when the positive operator $\sqrt A-\sqrt{A-B}$ vanishes, equivalently $B=0$. For $q>0$, $B$ is a sum of positive operators $qT_{ij}T_{ij}^\dagger$, where $T_{ij}=\sqrt{\rho_T}Q_{ij}\sqrt{\rho_T}$. Thus equality is equivalent to $T_{ij}=0$ for every pair, or Eq.~\eqref{eq:product-equality}. Equation~\eqref{eq:product-gap} gives the same condition for overlap equality.
\end{proof}

The gap measures the commutator component within the initial target support. A component that carries the state into the orthogonal complement has zero compression and can change the joint output while saturating the fidelity bound.

The compression also has a direct measurement interpretation. Initialize the switch in $\ket+_C\ket m_T$. The antisymmetric control outcome has probability~\cite{Gao2023}
\begin{equation}
 p_- =\frac14\sum_{ij}\bra m Q_{ij}^\dagger Q_{ij}\ket m.
\end{equation}
Resolving the target output into $\ket m$ and its orthogonal complement gives
\begin{align}
 p_{-,m}&=\frac14\Gamma(m)=\frac14\sum_{ij}|\bra m Q_{ij}\ket m|^2,\nonumber\\
 p_- -p_{-,m}&=\frac14\sum_{ij}
 \left\|(I-\ket m\bra m)Q_{ij}\ket m\right\|^2.
 \label{eq:commutator-return}
\end{align}
For a balanced control and a pure target, the loss in squared fidelity therefore equals the joint probability of the antisymmetric outcome and return to the initial target state. This return component also gives the gain in Theorem~\ref{thm:qubit}.

These quantities are independent of the Kraus representation. Under isometric changes $K'_\alpha=\sum_i U_{\alpha i}K_i$ and $L'_\beta=\sum_j V_{\beta j}L_j$, the family $Q_{ij}$ transforms by the isometry $U\otimes V$. The squared matrix element sums and $\mathcal D$ are therefore invariant. Zero padding accommodates representations of different sizes. Process tomography of the channels~\cite{Chuang1997} determines these invariants; Eq.~\eqref{eq:commutator-return} measures their return component directly.

\subsection{Orthogonal order branches and the qubit sign reversal}

Consider next
\begin{equation}
 \ket\Psi=a\ket0_C\ket m_T+b\ket1_C\ket n_T,
 \qquad \braket{m|n}=0.
 \label{eq:aligned}
\end{equation}
For $ab\neq0$, Eq.~\eqref{eq:aligned} aligns the control Schmidt basis~\cite{Horodecki2009} with the order basis: the two branches carry orthogonal target states. Equivalently, the reduced control state is diagonal in the order basis. Every maximally entangled state with a qubit control has this form. A nonmaximally entangled state instead selects a preferred control Schmidt axis.

For any $d\geq2$, introduce
\begin{align}
 x_{ij}&=\bra m K_iL_j\ket m,&
 y_{ij}&=\bra n K_iL_j\ket n,\nonumber\\
 u_{ij}&=\bra m L_jK_i\ket m,&
 v_{ij}&=\bra n L_jK_i\ket n,
 \label{eq:elements}
\end{align}
and set
\begin{equation}
\begin{aligned}
 P_m&=\sum_{ij}|x_{ij}|^2=\sum_{ij}|u_{ij}|^2,
 &P_n&=\sum_{ij}|y_{ij}|^2=\sum_{ij}|v_{ij}|^2,\\
 C&=\sum_{ij}x_{ij}y_{ij}^*,
 &D&=\sum_{ij}x_{ij}v_{ij}^*.
\end{aligned}
 \label{eq:constants}
\end{equation}
Channel commutativity gives the norm equalities. For $q=|a|^2|b|^2$, the joint Kraus expectations yield
\begin{align}
 F_{\FO}^2&=|a|^4P_m+|b|^4P_n+2q\operatorname{Re}C,\nonumber\\
 F_{\SW}^2&=|a|^4P_m+|b|^4P_n+2q\operatorname{Re}D.
 \label{eq:aligned-fidelities}
\end{align}
The relative phase of $a$ and $b$ cancels. The sign of the gap is determined by $\operatorname{Re}(D-C)$.

\begin{theorem}[Fidelity gap for aligned pure states of two qubits]\label{thm:qubit}
Let $\Phi$ and $\Psi$ be commuting CPTP maps on a qubit target. For the joint maps~\eqref{eq:switch}--\eqref{eq:fixed} and a normalized pure input~\eqref{eq:aligned} with $m,n$ an orthonormal target basis and $q=|a|^2|b|^2$,
\begin{equation}
 \ F_{\SW}^2-F_{\FO}^2
 =q\,\Gamma(m),\qquad
 \Gamma(m)=\sum_{ij}|\bra m Q_{ij}\ket m|^2.\ 
 \label{eq:qubit-gap}
\end{equation}
Consequently $F_{\SW}\geq F_{\FO}$. If $q>0$, equality holds exactly when $\bra m Q_{ij}\ket m=0$ for every $i,j$. The same condition then holds for $n$.
\end{theorem}

\begin{proof}
Because $m,n$ form a complete qubit basis, cyclicity of the trace gives
\begin{equation}
 x_{ij}+y_{ij}=u_{ij}+v_{ij},\qquad
 v_{ij}-y_{ij}=x_{ij}-u_{ij}.
 \label{eq:trace-identity}
\end{equation}
Together with the equal norms in Eq.~\eqref{eq:constants}, this gives
\begin{equation}
 \operatorname{Re}(D-C)
 =\operatorname{Re}\sum_{ij}x_{ij}(x_{ij}-u_{ij})^*
 =\frac12\sum_{ij}|x_{ij}-u_{ij}|^2.
 \label{eq:square-identity}
\end{equation}
Substitution into Eq.~\eqref{eq:aligned-fidelities} proves the gap. Its nonnegative summands give the equality condition; $\tr Q_{ij}=0$ gives the same condition for $n$.
\end{proof}

Equations~\eqref{eq:pure-product-gap} and~\eqref{eq:qubit-gap} differ only in sign. Product branches sample the same commutator expectation; aligned qubit branches sample opposite expectations. Coupling this amplitude to $Z_C$ in Eq.~\eqref{eq:order-amplitude} converts the product loss into the aligned gain.

For a qubit target, the overlap identity extends to every mixed state supported on
$\cH_{mn}=\operatorname{span}\{\ket{0m},\ket{1n}\}$:
\begin{equation}
 \rho=\begin{pmatrix}r&c\\c^*&s\end{pmatrix}_{\!\cH_{mn}},
 \qquad r+s=1,\quad |c|^2\leq rs.
 \label{eq:mixed-input}
\end{equation}
Its two overlaps have the form in Eq.~\eqref{eq:aligned-fidelities}, with
$|a|^4,|b|^4,q$ replaced by $r^2,s^2,|c|^2$. Thus
\begin{equation}
 P_{\SW}-P_{\FO}=|c|^2\Gamma(m).
 \label{eq:mixed-gap}
\end{equation}
If $c=0$, the two outputs coincide. For $c\neq0$, overlap equality requires $\Gamma(m)=0$. The relevant coherence~\cite{Baumgratz2014} is the amplitude $c$ in the joint basis $\{\ket{0m},\ket{1n}\}$; the reduced control state is diagonal. Appendix~\ref{app:mixed} gives the corresponding root fidelities.

\section{Commutator geometry and the choice of input}\label{sec:geometry}

\subsection{The optimal target basis}

For a qubit, every commutator is traceless and has a Pauli expansion
\begin{equation}
 Q_{ij}=\bm d_{ij}\cdot\bm\sigma,\qquad
 d_{ij,\alpha}=\frac12\tr(\sigma_\alpha Q_{ij}).
\end{equation}
Define the real symmetric matrix
\begin{equation}
 G_{\alpha\beta}=\operatorname{Re}\sum_{ij}
 d_{ij,\alpha}^*d_{ij,\beta},\qquad \alpha,\beta\in\{x,y,z\}.
 \label{eq:G}
\end{equation}
For a pure target state with Bloch vector $\bm n$,
\begin{equation}
 \Gamma(m)=\bm n^{\mathsf T}G\bm n,\qquad |\bm n|=1.
 \label{eq:bloch}
\end{equation}
Thus $G$ is positive semidefinite and independent of the Kraus representation. Under a target basis rotation represented by $O\in\mathrm{SO}(3)$, its coordinate matrix transforms as $G\mapsto OGO^{\mathsf T}$. Its eigenvectors identify the extremal target directions.

\begin{corollary}[Basis optimization at fixed entanglement]\label{cor:optimization}
For the commuting qubit channels of Theorem~\ref{thm:qubit}, let $\mathcal C=2|ab|$ be the concurrence of the pure state~\eqref{eq:aligned}~\cite{Wootters1998}. Among such inputs with fixed $\mathcal C$,
\begin{equation}
 \frac{\mathcal C^2}{4}\lambda_{\min}(G)
 \leq F_{\SW}^2-F_{\FO}^2
 \leq\frac{\mathcal C^2}{4}\lambda_{\max}(G).
 \label{eq:optimal-gap}
\end{equation}
Both endpoints are attained by choosing the Bloch axis of $m,n$ along the corresponding eigenvector. Maximizing also over the Schmidt coefficients gives $\lambda_{\max}(G)/4$, attained by a maximally entangled input in the optimal basis.
\end{corollary}

\begin{proof}
Use $q=\mathcal C^2/4$ in Eq.~\eqref{eq:qubit-gap} and apply the extremal property of the Rayleigh quotient to Eq.~\eqref{eq:bloch}.
\end{proof}

Within the aligned family, concurrence sets the prefactor and $G$ selects the target direction. The same directions govern the product loss for a fixed pure control. A zero eigenvalue gives equality in both comparisons. Equation~\eqref{eq:optimal-gap} optimizes the gap in squared fidelity; the corresponding root fidelity gap also depends on the baseline fidelity.

\subsection{Changing the sign at fixed entanglement}

The distinction between entanglement and alignment is explicit for the commuting bit flip and Pauli $Y$ channels
\begin{equation}
 \Phi_\gamma(\rho)=(1-\gamma)\rho+\gamma X\rho X,\qquad
 \Psi_\gamma(\rho)=(1-\gamma)\rho+\gamma Y\rho Y,
 \qquad 0\leq\gamma\leq1.
 \label{eq:pauli}
\end{equation}
Their only nonzero Kraus commutator is $Q_{11}=2i\gamma Z$, so
$G=4\gamma^2\operatorname{diag}(0,0,1)$. The computational target basis maximizes the aligned gap, while every equatorial basis gives equality. Figure~\ref{fig:orientation}(a) illustrates the opposite gaps for a balanced product input and an aligned input with equal Schmidt coefficients as the target basis varies.

More generally, the two joint channels in this example differ only in the branch where both flips occur. Its operator is $iI_C\otimes Z_T$ in the fixed order and $iZ_C\otimes Z_T$ in the switch. For \emph{every pure input} $\ket\psi$,
\begin{equation}
 \ \Delta P=\gamma^2\left[
 \langle Z_C\otimes Z_T\rangle_\psi^2-
 \langle I_C\otimes Z_T\rangle_\psi^2\right].\ 
 \label{eq:pauli-general}
\end{equation}
If $\ket\psi=\sqrt p\ket0\ket u+\sqrt{1-p}\ket1\ket v$ with normalized, not necessarily orthogonal $u,v$, this reads
\begin{equation}
 \Delta P=-4p(1-p)\gamma^2 z_u z_v,
 \qquad z_u=\bra u Z\ket u,\quad z_v=\bra v Z\ket v.
 \label{eq:pauli-conditional}
\end{equation}
For $0<p<1$ and $\gamma>0$, the gap is positive when the conditional $Z$ polarizations have opposite signs and negative when they have the same nonzero sign. Orthogonal qubit states have opposite polarizations; linearly independent conditional states can occupy the same polarization hemisphere.

To vary the control basis without changing entanglement, consider
\begin{equation}
 \ket{\Psi_{\theta,\alpha}}=
 [R_y(\alpha)_C\otimes I_T]
 \left(\cos\frac\theta2\ket{00}+\sin\frac\theta2\ket{11}\right),
 \qquad R_y(\alpha)=e^{-i\alpha Y/2}.
 \label{eq:rotated-input}
\end{equation}
Its concurrence is $\sin\theta$ for $0\leq\theta\leq\pi$, independent of $\alpha$, whereas
\begin{equation}
 \Delta P=\gamma^2(\cos^2\alpha-\cos^2\theta).
 \label{eq:orientation-gap}
\end{equation}
At $\alpha=0$ the gain is $\gamma^2\sin^2\theta$. At $\alpha=\pi/2$ it is $-\gamma^2\cos^2\theta$. For $0<\sin\theta<1$ and $\gamma>0$, control rotations therefore produce both signs at fixed concurrence. The maximally entangled case has a nonnegative gap for all $\alpha$, consistent with its diagonal control marginal in every basis. Figure~\ref{fig:orientation}(b) displays this orientation dependence using concurrence $\mathcal C=\sin\theta$ as the horizontal coordinate. For $0\leq\theta,\alpha\leq\pi/2$, the equality curve is $\alpha=\arcsin\mathcal C$; the marked points at $\mathcal C=\sqrt3/2$ illustrate gain, equality, and loss at fixed entanglement.

\begin{figure}[htbp]
 \centering
 \includegraphics[width=119mm]{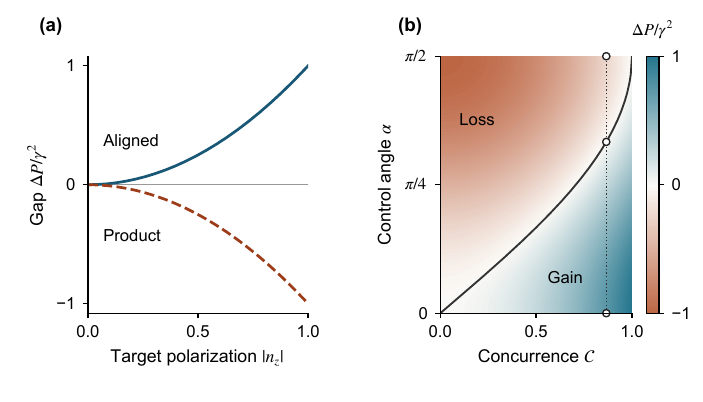}
 \caption{Commutator contributions and input orientation for Pauli $X$ and $Y$ channels with a common error probability $0<\gamma\leq1$. The gap is $\Delta P=F_{\mathrm{SW}}^2-F_{\mathrm{FO}}^2$. (a)~For the product input $|+\rangle_C|m\rangle_T$ and the aligned input $(|0m\rangle+|1n\rangle)/\sqrt2$, with $\langle m|n\rangle=0$ and $n_z=\langle m|Z|m\rangle$, the normalized gaps are $-n_z^2$ (dashed) and $+n_z^2$ (solid), respectively. Both vanish for an equatorial target basis and have maximal magnitude along the $Z$ axis. (b)~For $[R_y(\alpha)\otimes I][\cos(\theta/2)|00\rangle+\sin(\theta/2)|11\rangle]$, with $0\leq\theta,\alpha\leq\pi/2$, the concurrence is $\mathcal C=\sin\theta$ and the normalized gap is $\mathcal C^2-\sin^2\alpha$. The solid curve $\alpha=\arcsin\mathcal C$ separates gain and loss. The dotted line fixes $\mathcal C=\sqrt3/2$; the circles at $\alpha=0,\pi/3,\pi/2$ give normalized gaps $3/4,0,-1/4$, respectively}
 \label{fig:orientation}
\end{figure}

A concrete entangled input with a negative gap is
\begin{equation}
 \ket\chi=\frac{\sqrt3}{2}\ket+_C\ket0_T
 +\frac12\ket-_C\ket1_T.
 \label{eq:qubit-counterexample}
\end{equation}
It has concurrence $\sqrt3/2$, yet
\begin{equation}
 P_{\SW}=(1-\gamma)^2,
 \qquad P_{\FO}=(1-\gamma)^2+\frac{\gamma^2}{4}.
 \label{eq:counterexample-gap}
\end{equation}
Thus $F_{\SW}<F_{\FO}$ for every $\gamma>0$, with $F_{\SW}=0$ and $F_{\FO}=1/2$ at $\gamma=1$. The counterexample and Eq.~\eqref{eq:orientation-gap} identify the control basis as part of the physical input specification.

\section{The role of target dimension}\label{sec:dimension}

The two conditional target states in the qubit theorem form a complete basis. In higher dimension, the remaining diagonal entries also contribute to the zero trace of each commutator. A qutrit permits equal commutator expectations on an orthogonal pair and reverses the sign of the aligned input gap.

\begin{proposition}[A commuting qutrit counterexample]\label{prop:qutrit}
There exist commuting unitary qutrit channels and an input of the form~\eqref{eq:aligned} with equal Schmidt coefficients for which
\begin{equation}
 F_{\SW}=\frac14<\frac12=F_{\FO}.
 \label{eq:qutrit-values}
\end{equation}
\end{proposition}

\begin{proof}
Let $\omega=e^{2\pi i/3}$ and define
\begin{equation}
 X_3\ket j=\ket{j+1\pmod3},\qquad
 Z_3=\operatorname{diag}(1,\omega,\omega^2),\qquad
 K=X_3,\quad L=X_3^\dagger Z_3.
 \label{eq:qutrit-operators}
\end{equation}
Then $KL=Z_3$ and $LK=\omega Z_3$, so the two unitary channels commute and their common composition is $\Ad_{Z_3}$. Choose
\begin{equation}
 \ket m=\frac{\ket0+\ket1}{\sqrt2},\qquad
 \ket n=\frac{\ket0-\ket1}{\sqrt2},\qquad
 \ket\Psi=\frac{\ket0_C\ket m+\ket1_C\ket n}{\sqrt2}.
 \label{eq:qutrit-input}
\end{equation}
The orthogonal states $m,n$ have the same diagonal expectation
$z=\bra m Z_3\ket m=\bra n Z_3\ket n=(1+\omega)/2$. The fixed-order return amplitude is $z$, and the switch amplitude is $(1+\omega)z/2$. Since $|1+\omega|=1$, their moduli are $1/2$ and $1/4$.
\end{proof}

In this example
\begin{equation}
 [K,L]=(1-\omega)Z_3,\qquad
 \bra m[K,L]\ket m=\bra n[K,L]\ket n=\frac{1-\omega^2}{2}.
 \label{eq:qutrit-commutator}
\end{equation}
The third target direction balances the equal expectations on $m,n$ to give zero trace. The qubit proof uses completeness of the pair precisely at this step. Here the switched unitary factors as $(\ket0\bra0+\omega\ket1\bra1)_C\otimes Z_3$: the fidelity change comes from a relative control phase.

The violation also occurs for noisy unital channels. For
\begin{equation}
 \Phi_t=(1-t)\Id+t\Ad_K,\qquad
 \Psi_t=(1-t)\Id+t\Ad_L,
 \qquad 0\leq t\leq1,
 \label{eq:noisy-qutrit}
\end{equation}
the two channels commute. All branches except the simultaneous $K,L$ branch are identical in the switch and fixed order. The input~\eqref{eq:qutrit-input} therefore has
\begin{equation}
 \Delta P=-\frac{3t^2}{16}<0\qquad (t>0).
 \label{eq:noisy-qutrit-gap}
\end{equation}
Thus every $t>0$ in this noisy unital family gives a negative gap.

The accessible target sector determines the relevant dimension. Suppose there is a two-dimensional subspace $\mathcal V\subseteq\cH_T$ invariant under every $K_i$ and $L_j$. Their restrictions are CPTP maps on $\mathcal V$ and commute there. Theorem~\ref{thm:qubit} then applies to any orthonormal $m,n\in\mathcal V$, even if the ambient space is larger. The trace argument therefore applies to a common invariant qubit sector. In the qutrit construction, the individual channels carry the chosen target support through a larger sector, even though their composition preserves that support.

\section{Encoded state preservation and parameter estimation}
\label{sec:encoding}

\subsection{From a logical phase error to a stabilizer}

Logical error structure organizes the preservation of encoded information in quantum error correction~\cite{Preskill1998,Gottesman1997}. For the equally noisy Pauli channels~\eqref{eq:pauli}, the positive fidelity gap has a direct interpretation in the subspace
\begin{equation}
 \cH_L=\operatorname{span}\{\ket{0_L}=\ket{00},
 \ket{1_L}=\ket{11}\},\qquad
 \Pi_L=\frac{I_C\otimes I_T+Z_C\otimes Z_T}{2}.
 \label{eq:code}
\end{equation}
Here $Z_C\otimes Z_T$ acts as a stabilizer~\cite{Gottesman1997}, while $I_C\otimes Z_T$ acts as the logical Pauli operator $Z_L$. The switch therefore converts the branch with both flips from a logical phase error into the identity on the code. Table~\ref{tab:branches} summarizes the four error branches, and Fig.~\ref{fig:qfi}(a) illustrates the different logical actions of the double error branch.

\begin{table}[htbp]
\caption{Action of the Pauli branches on $\cH_L$. Overall phases are omitted. A single flip maps the state into the orthogonal odd parity subspace}
\label{tab:branches}
\centering
\begin{tabular}{@{}llll@{}}
\toprule
Branch & Probability & Fixed order & Switch\\
\midrule
No flip & $(1-\gamma)^2$ & $I_L$ & $I_L$\\
$X$ only & $\gamma(1-\gamma)$ & Outside $\cH_L$ & Outside $\cH_L$\\
$Y$ only & $\gamma(1-\gamma)$ & Outside $\cH_L$ & Outside $\cH_L$\\
Both flips & $\gamma^2$ & $Z_L$ & $I_L$\\
\bottomrule
\end{tabular}
\end{table}

Write $A_\gamma=(1-\gamma)^2+\gamma^2$. For every density operator $\rho_L$ supported on $\cH_L$, Table~\ref{tab:branches} gives
\begin{align}
 \Pi_L\cS(\rho_L)\Pi_L&=A_\gamma\rho_L,
 \label{eq:code-switch}\\
 \Pi_L\mathcal F(\rho_L)\Pi_L
 &=(1-\gamma)^2\rho_L+\gamma^2Z_L\rho_L Z_L.
 \label{eq:code-fixed}
\end{align}
Consequently,
\begin{equation}
 F(\rho_L,\cS(\rho_L))=\sqrt{A_\gamma},
 \label{eq:code-fidelity}
\end{equation}
independently of the logical state and its purity. The same statement holds if the encoded qubit is correlated with an untouched reference $R$: for every state $\rho_{RL}$ supported on $\cH_R\otimes\cH_L$,
\begin{equation}
 F\!\left(\rho_{RL},(\Id_R\otimes\cS)(\rho_{RL})\right)
 =\sqrt{A_\gamma}.
 \label{eq:reference-fidelity}
\end{equation}
Compressing the output onto this support gives $A_\gamma\rho_{RL}$, so the operator under the fidelity square root is $A_\gamma\rho_{RL}^2$. Equation~\eqref{eq:reference-fidelity} follows from $\tr\rho_{RL}=1$ and holds for every encoded state and its correlations with a reference.

For a pure logical state $a\ket{0_L}+b\ket{1_L}$, Eq.~\eqref{eq:code-fixed} gives
\begin{equation}
 F_{\SW}^2=A_\gamma,\qquad
 F_{\FO}^2=A_\gamma-4\gamma^2|a|^2|b|^2.
 \label{eq:logical-pure}
\end{equation}
The gap is largest for equal logical populations because that state is most sensitive to the logical phase flip. The switch removes this logical error while leaving the single flip leakage unchanged. At $\gamma=1$, both physical channels are deterministic unitary flips; the switch preserves the whole code subspace, whereas the fixed order applies $Z_L$. For example, it maps $\ket{\Phi^+}$ to the orthogonal pure state $\ket{\Phi^-}$.

For the mixed logical state~\eqref{eq:mixed-input}, the overlap gain depends only on its coherence magnitude:
\begin{equation}
 \Delta P=4\gamma^2|c|^2.
 \label{eq:logical-mixed}
\end{equation}
A fixed phase rotation of $c$ preserves the gap. Dephasing the preparation as $c\mapsto\kappa c$ multiplies it by $|\kappa|^2$, with fidelity and overlap evaluated relative to that preparation. The magnitude of the logical coherence therefore sets the overlap gain; its phase fixes the azimuthal direction of the logical Bloch vector.

\subsection{A channel ordering for encoded information}
\label{sec:encoded-qfi}

The same stabilizer structure gives a channel relation between the outputs. Let $\Pi_e=\Pi_L$, $\Pi_o=I-\Pi_L$, and write
$B_\gamma=1-2\gamma$. For an encoded state
$\rho_L=\left(\begin{smallmatrix}r&c\\c^*&s\end{smallmatrix}\right)$,
the even parity blocks in Eqs.~\eqref{eq:code-switch} and
\eqref{eq:code-fixed} are
\begin{equation}
 B_{\SW}^{(e)}=A_\gamma\rho_L,\qquad
 B_{\FO}^{(e)}=
 \begin{pmatrix}A_\gamma r&B_\gamma c\\B_\gamma c^*&A_\gamma s\end{pmatrix}.
 \label{eq:qfi-even-blocks}
\end{equation}
Both odd parity blocks equal
$2\gamma(1-\gamma)\operatorname{diag}(r,s)$ in the basis
$\ket{01},\ket{10}$. The full outputs are direct sums of these parity blocks.

\begin{theorem}[Channel ordering on the Pauli code]\label{thm:degradation}
For the Pauli channels~\eqref{eq:pauli} with a common noise probability $0\leq\gamma\leq1$ and the code $\cH_L$ in Eq.~\eqref{eq:code}, define a CPTP map
$\mathcal R_\gamma$ by
\begin{equation}
 R_0=\frac{1-\gamma}{\sqrt{A_\gamma}}\Pi_e,\qquad
 R_1=\frac{\gamma}{\sqrt{A_\gamma}}(I_C\otimes Z_T)\Pi_e,\qquad
 R_2=\Pi_o.
 \label{eq:qfi-degradation-kraus}
\end{equation}
For every density operator supported on $\cH_L$,
\begin{equation}
 \ \mathcal F(\rho_L)=\mathcal R_\gamma\bigl(\cS(\rho_L)\bigr).\ 
 \label{eq:qfi-degradation}
\end{equation}
The identity holds with an arbitrary untouched reference system. At each fixed $\gamma$, $\mathcal R_\gamma$ is independent of every parameter encoded in the input state.
\end{theorem}

\begin{proof}
The operators satisfy $\sum_aR_a^\dagger R_a=\Pi_e+\Pi_o=I$. On the even sector, $\mathcal R_\gamma$ preserves populations and multiplies coherences by $B_\gamma/A_\gamma$, converting $B_{\SW}^{(e)}$ into $B_{\FO}^{(e)}$. It leaves the common odd block unchanged. This proves Eq.~\eqref{eq:qfi-degradation}. Replacing scalar entries by operators on a reference gives the same identity.
\end{proof}

Thus conditional logical dephasing converts the switched output into the output under fixed order. For a differentiable family
$\rho(\boldsymbol\lambda)$, define the symmetric logarithmic derivatives
and the quantum Fisher information (QFI) matrix by~\cite{Braunstein1994}
\begin{equation}
 \partial_\mu\rho=\frac{L_\mu\rho+\rho L_\mu}{2},\qquad
 J_{\mu\nu}=\operatorname{Re}\tr(\rho L_\mu L_\nu).
 \label{eq:qfi-definition}
\end{equation}
For one parameter, $J_{\lambda\lambda}$ is the maximum classical Fisher information over output measurements. At fixed $\gamma$, the map $\mathcal R_\gamma$ is independent of the parameters in $\rho_L(\boldsymbol\lambda)$. Applying QFI monotonicity in every parameter direction gives
\begin{equation}
 \boldsymbol J_{\SW}(\boldsymbol\lambda)
 \succeq\boldsymbol J_{\FO}(\boldsymbol\lambda).
 \label{eq:qfi-code-ordering}
\end{equation}
This ordering includes mixed encoded states and correlations with a reference. The same channel relation also orders trace distance and optimal binary discrimination.

\subsection{Encoded phase and noise information}
\label{sec:phase-noise}

Consider a phase written onto the target before the two channels:
\begin{equation}
 \ket{\psi_{\theta,\varphi}}=
 \cos\frac\theta2\ket{00}+e^{i\varphi}\sin\frac\theta2\ket{11}.
 \label{eq:qfi-input}
\end{equation}
The encoding operation is
$U_\varphi=\ket0\bra0+e^{i\varphi}\ket1\bra1$ on $T$.
For fixed $0<\gamma<1$ and $0<\theta<\pi$, the QFI matrices in coordinates
$(\theta,\varphi)$ are
\begin{equation}
 \boldsymbol J_{\SW}=\operatorname{diag}
 \left(1,A_\gamma\sin^2\theta\right),\qquad
 \boldsymbol J_{\FO}=\operatorname{diag}
 \left(1,\frac{B_\gamma^2}{A_\gamma}\sin^2\theta\right).
 \label{eq:qfi-phase-metrics}
\end{equation}
Appendix~\ref{app:qfi} derives Eq.~\eqref{eq:qfi-phase-metrics} from the parity blocks. The populations retain the polar information. The even sector carries the phase information with probability $A_\gamma$, so its QFI is weighted by that probability in the full output.

For this encoded family, the commutator fidelity gap and phase information
gain satisfy the exact relation
\begin{equation}
 \ J_{\varphi,\SW}-J_{\varphi,\FO}
 =\frac{4\gamma^2(1-\gamma)^2}{A_\gamma}\sin^2\theta
 =\frac{4(1-\gamma)^2}{A_\gamma}\Delta P.\ 
 \label{eq:qfi-fidelity-link}
\end{equation}
At $\gamma=1/2$, the balanced input has phase QFI $1/2$ under the switch and zero under fixed order. A parity measurement followed by the locally optimal phase measurement in the even sector attains the switched value: that sector occurs with probability $1/2$ and has conditional phase QFI one.

Fidelity with the input and QFI describe different aspects of the output. At $\gamma=1$,
the balanced input has squared fidelities $P_{\FO}=0$ and $P_{\SW}=1$,
while both phase QFIs equal one. The fixed-order logical $Z$ maps the
input to an orthogonal state and preserves the distances between
neighboring phases. In a regular constant-rank family, the corresponding
root fidelity expansion expresses the local Bures geometry~\cite{Bures1969,Braunstein1994}
\begin{equation}
 F(\rho_\lambda,\rho_{\lambda+d\lambda})
 =1-\frac18J_{\lambda\lambda}\,d\lambda^2+o(d\lambda^2).
 \label{eq:qfi-local-fidelity}
\end{equation}
The channel ordering establishes that the switched output has at least as much local distinguishability as the output under fixed order for encoded state families. Appendix~\ref{app:qfi-separation} gives an aligned example with amplitude damping in which input fidelity increases and phase QFI decreases.

Figure~\ref{fig:qfi} connects the logical error mechanism in panel (a) with the squared input fidelity and phase QFI of the balanced encoded probe in panels (b,c). Panel (d) compares its normalized noise QFI under the two joint channels and shows the optimum attained by a product probe under the switch.

\begin{figure}[htbp]
 \centering
 \includegraphics[width=119mm]{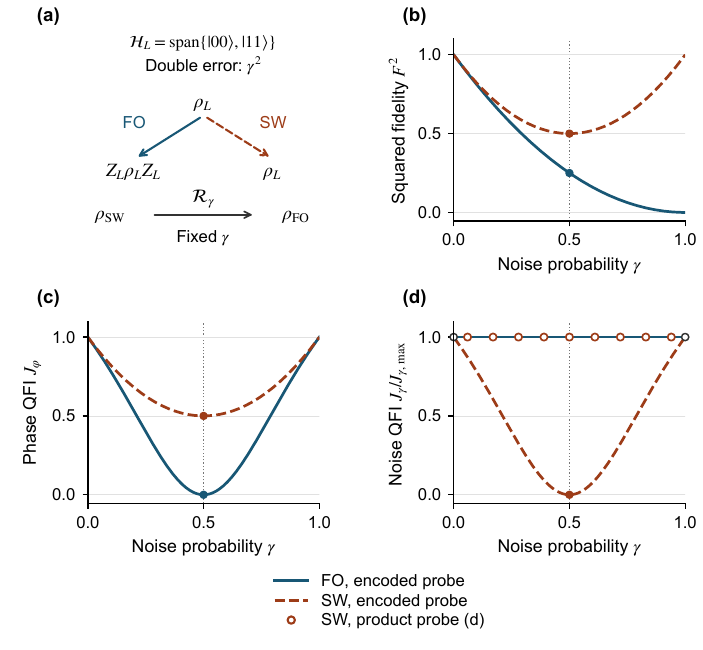}
 \caption{Logical error structure, state preservation, and parameter estimation for Pauli $X$ and $Y$ channels with a common error probability $\gamma$. (a)~The double error branch, of probability $\gamma^2$, acts on a state $\rho_L$ in $\cH_L=\operatorname{span}\{\ket{00},\ket{11}\}$ as $\rho_L\mapsto Z_L\rho_L Z_L$ under fixed order (FO) and as $\rho_L\mapsto\rho_L$ under the quantum switch (SW), where $Z_L$ is the logical Pauli $Z$ operator. The branch outputs shown are normalized. The lower arrow represents the CPTP map $\mathcal R_\gamma$, which converts the complete switched output into the fixed order output for every encoded state at fixed $\gamma$. (b,c)~Squared input fidelity and phase QFI for $\ket{\psi_\varphi}=(\ket{00}+e^{i\varphi}\ket{11})/\sqrt2$, with the phase encoded before the channels. (d)~Noise QFI normalized by $J_{\gamma,\max}=2/[\gamma(1-\gamma)]$, the maximum over inputs to either joint map. Solid and dashed curves use the same encoded probe under FO and SW. Open circles within the interval show analytical values for the SW product probe $\ket+_C\ket0_T$, which attains the same optimum as the FO encoded probe. Black open endpoint circles denote limits as $\gamma\to0,1$. Dotted vertical lines mark $\gamma=1/2$. All quantities in panels (b--d) refer to the complete joint output}
 \label{fig:qfi}
\end{figure}

To estimate the noise probability, hold $\theta,\varphi$ fixed and vary $\gamma$. With
$w_\gamma=\gamma(1-\gamma)$, the exact QFIs for $0<\gamma<1$ are
\begin{equation}
 J_{\gamma,\SW}=\frac{2B_\gamma^2}{A_\gamma w_\gamma},\qquad
 J_{\gamma,\FO}=J_{\gamma,\SW}+\frac{4\sin^2\theta}{A_\gamma}.
 \label{eq:qfi-noise}
\end{equation}
The switch merges the histories with zero and two flips inside the code, preserving the encoded signal while reducing information about the error history. Its encoded output is symmetric under $\gamma\leftrightarrow1-\gamma$, so its first derivative vanishes and $J_{\gamma,\SW}=0$ at $\gamma=1/2$. The balanced probe under fixed order has $J_{\gamma,\FO}=8$ there. Here $\gamma$ enters the degradation map itself; the signal ordering in Eq.~\eqref{eq:qfi-code-ordering} instead holds at fixed $\gamma$.

For the joint maps~\eqref{eq:switch} and~\eqref{eq:fixed} with Pauli noise~\eqref{eq:pauli}, optimizing over inputs independent of $\gamma$ gives, for $0<\gamma<1$,
\begin{equation}
 \max_{\rho_{CT}}J_{\gamma,\FO}
 =\max_{\rho_{CT}}J_{\gamma,\SW}
 =\frac{2}{\gamma(1-\gamma)}.
 \label{eq:qfi-noise-optimal}
\end{equation}
A Bell input attains the bound under fixed order, and $\ket+_C\ket0_T$ attains it under the switch. Appendix~\ref{app:qfi} derives the bound from the two Bernoulli error labels and gives the attaining measurements. Thus the encoded probe preserves phase information and identifies noise differently, while both joint channels reach the same optimum for noise estimation.

\FloatBarrier
\section{Conclusion}\label{sec:conclusion}

Exact Kraus commutator identities characterize input preservation
under fixed order and the quantum switch for commuting channels.
For every product input, the fidelity under fixed order is at least
as high as under the switch. For product inputs with a pure control,
commutators compressed onto the target support determine the exact
overlap loss in any finite target dimension, with a common equality
condition for overlap and fidelity. The reverse fidelity inequality
holds for pure states of two qubits whose control Schmidt basis can
be chosen to coincide with the order basis. For these aligned inputs,
the squared fidelity gain is a sum of squared moduli of commutator
expectation values and is therefore nonnegative.

The commutator matrix describes the geometry of this gain. An
eigenvector associated with its largest eigenvalue determines a
target basis that maximizes the squared fidelity gap at fixed
concurrence within the aligned family. In the Pauli example, control
rotations produce gaps of opposite sign at the same nonzero,
nonmaximal concurrence. The qutrit counterexample shows the role
of accessible target dimension, while aligned inputs supported on
a common invariant qubit sector retain the nonnegative gap.
Thus, both input orientation and accessible target dimension
affect the fidelity comparison.

For Pauli $X$ and $Y$ channels with equal error probabilities, the
switch replaces a logical phase flip in the double error branch
with an operator that stabilizes the joint code. The fidelity
between the switched output and the input is independent of the
encoded state, even when the input is correlated with an untouched
reference. Conditional logical dephasing maps the switched output
to its fixed order counterpart. At fixed noise, this relation shows
that the switch preserves at least as much distinguishability and
QFI as fixed order for every encoded state family. The QFI ordering
also holds for multiparameter families in the positive semidefinite
matrix sense. For the pure encoded family considered here, an exact
relation connects the phase QFI gain to the squared input fidelity gap.

Under the switch, the branches with no error and with both errors
act identically on the code, making these two histories
indistinguishable in the output. The exact QFI expressions quantify
the phase information gain and the reduction in information about
the noise probability for the same pure encoded probe. Optimizing
the probes yields the same maximal noise QFI for both joint channels,
attained by a Bell probe under fixed order and a product probe under
the switch. The amplitude damping example illustrates a different
relation between the two measures: input fidelity increases while
phase QFI decreases. Fidelity measures how well a given input state
is preserved, whereas QFI quantifies local distinguishability
within a parameterized state family. The commutator identities
characterize the fidelity comparison, while the channel relation
on the Pauli code establishes the ordering of encoded information
preservation.

\backmatter

\section*{Statements and Declarations}

\paragraph*{Funding}
This work was supported by Hebei University of Science and Technology
(Starting Grant No.~81/1181298).

\paragraph*{Competing interests}
The authors declare that they have no competing interests.

\paragraph*{Author contributions}
Both authors contributed to the conceptualization, methodology,
and manuscript preparation.

\begin{appendices}
\section{Mixed states on the aligned subspace}\label{app:mixed}

Let $\rho$ have the form~\eqref{eq:mixed-input}, and let $m,n$ be orthonormal in a target space of any dimension $d\geq2$. Use the quantities in Eqs.~\eqref{eq:elements}--\eqref{eq:constants} and write $t_{\FO}=C$, $t_{\SW}=D$. Compression of the two outputs onto $\cH_{mn}$ gives
\begin{equation}
 B_X=\begin{pmatrix}
 rP_m & c\,t_X\\ c^*t_X^*&sP_n
 \end{pmatrix},\qquad X\in\{\FO,\SW\}.
 \label{eq:compressed-mixed}
\end{equation}
The off-diagonal entry is $cC$ under fixed order and $cD$ under the switch, directly from Eq.~\eqref{eq:constants}.
Tracing against the input gives
\begin{equation}
 P_X=r^2P_m+s^2P_n+2|c|^2\operatorname{Re}t_X.
 \label{eq:mixed-overlaps}
\end{equation}
For a qubit target, Eq.~\eqref{eq:square-identity} reduces the difference to Eq.~\eqref{eq:mixed-gap}.

Only the compressed output $B_X$ enters the fidelity. For any positive $2\times2$ matrix $A$, $(\tr\sqrt A)^2=\tr A+2\sqrt{\det A}$. Setting $A=\sqrt\rho B_X\sqrt\rho$ gives
\begin{equation}
\ F_X^2=P_X+
 2\sqrt{(rs-|c|^2)\bigl(rsP_mP_n-|c|^2|t_X|^2\bigr)}.\ 
 \label{eq:mixed-fidelity}
\end{equation}
The positivity of $\rho$ and $B_X$ ensures that both factors under the square root are nonnegative. The formula remains valid at rank-deficient endpoints by continuity, and becomes $F_X^2=P_X$ for a pure input.

The overlap difference depends on $\operatorname{Re}(D-C)$. Root fidelity also depends on $|C|$ and $|D|$ through the determinant term, which records the spectrum of the compressed output.

For the Pauli code of Sec.~\ref{sec:encoding}, $P_m=P_n=D=A_\gamma$ and $C=B_\gamma:=1-2\gamma$ are real. Equations~\eqref{eq:mixed-overlaps} and~\eqref{eq:mixed-fidelity} become
\begin{align}
 F_{\SW}^2&=A_\gamma,\nonumber\\
 F_{\FO}^2&=A_\gamma(r^2+s^2)+2|c|^2 B_\gamma
 +2\sqrt{(rs-|c|^2)(rsA_\gamma^2-|c|^2B_\gamma^2)}.
 \label{eq:pauli-mixed-fidelity}
\end{align}
The normalized compression $B_{\FO}/A_\gamma$ is a state, so $F_{\FO}\leq\sqrt{A_\gamma}$, with equality exactly when $B_{\FO}=A_\gamma\rho_L$. By Eq.~\eqref{eq:code-fixed}, this is equivalent to $Z_L\rho_L Z_L=\rho_L$, or $c=0$, for $\gamma>0$. At $\gamma=0$, every input gives equality.

\section{QFI of the encoded Pauli outputs}\label{app:qfi}

Let $\sigma=\bigoplus_k p_k\sigma_k$ on fixed, mutually orthogonal subspaces, with normalized $\sigma_k$ and $p_k>0$ on the parameter interval. The QFI decomposes as
\begin{equation}
 J_\lambda(\sigma)=\sum_k\frac{(\partial_\lambda p_k)^2}{p_k}
 +\sum_k p_kJ_\lambda(\sigma_k).
 \label{eq:qfi-block-rule}
\end{equation}
For a full-rank qubit with Bloch vector $\bm r$, the QFI matrix is
\begin{equation}
 J_{\mu\nu}=\partial_\mu\bm r\cdot\partial_\nu\bm r
 +\frac{(\bm r\cdot\partial_\mu\bm r)
 (\bm r\cdot\partial_\nu\bm r)}{1-|\bm r|^2}.
 \label{eq:qfi-bloch}
\end{equation}
For a pure family it is
$4\operatorname{Re}[\braket{\partial_\mu\psi|\partial_\nu\psi}
-\braket{\partial_\mu\psi|\psi}\braket{\psi|\partial_\nu\psi}]$.
Equation~\eqref{eq:qfi-definition} gives these formulas. Throughout this appendix, $0<\gamma<1$ unless an endpoint is stated explicitly.

\subsection{Signal parameters}
The normalized even block of Eq.~\eqref{eq:qfi-input} has Bloch vector
\begin{equation}
 \bm r_e=(\eta\sin\theta\cos\varphi,
 \eta\sin\theta\sin\varphi,\cos\theta),\qquad
 \eta_{\SW}=1,\quad\eta_{\FO}=\frac{B_\gamma}{A_\gamma}.
\end{equation}
For $0<\theta<\pi$ at fixed $\gamma$, the even block has polar QFI $1$, phase QFI $\eta^2\sin^2\theta$, and zero cross term. The pure state formula applies to SW and Eq.~\eqref{eq:qfi-bloch} to FO. The odd block is
$\operatorname{diag}(\cos^2(\theta/2),\sin^2(\theta/2))$, whose polar QFI
is one and phase QFI is zero. The block weights are $A_\gamma$ and
$2w_\gamma$, both independent of the signal parameters. Substitution
into Eq.~\eqref{eq:qfi-block-rule} gives
Eq.~\eqref{eq:qfi-phase-metrics} throughout $0<\theta<\pi$.
At the polar endpoints
$\theta=0,\pi$ and $0<\gamma<1$, direct SLD evaluation gives
$J_{\theta\theta,\SW}=A_\gamma$ and
$J_{\theta\theta,\FO}=B_\gamma^2/A_\gamma$. The interior polar QFI tends to one. The difference between this limit and the endpoint value reflects the change in rank~\cite{Safranek2017}.

\subsection{Noise probability}
For SW, the normalized parity blocks are independent of $\gamma$. With
$A_\gamma'=-2B_\gamma$ and $(2w_\gamma)'=2B_\gamma$, their weights give
\begin{equation}
 \frac{(A_\gamma')^2}{A_\gamma}
 +\frac{[(2w_\gamma)']^2}{2w_\gamma}
 =\frac{2B_\gamma^2}{A_\gamma w_\gamma}.
\end{equation}
For FO, the normalized even state also depends on $\gamma$ through
$\eta=B_\gamma/A_\gamma$. Since
\begin{equation}
 \eta'=-\frac{4w_\gamma}{A_\gamma^2},\qquad
 1-\eta^2=\frac{4w_\gamma^2}{A_\gamma^2},
\end{equation}
the additional even sector information is
\begin{equation}
 A_\gamma J_{\gamma,e}
 =A_\gamma\frac{(\eta')^2\sin^2\theta}{1-\eta^2}
 =\frac{4\sin^2\theta}{A_\gamma}.
\end{equation}
This proves Eq.~\eqref{eq:qfi-noise}. Combining it with Eq.~\eqref{eq:qfi-fidelity-link} relates the two QFI differences:
\begin{equation}
 J_{\varphi,\SW}-J_{\varphi,\FO}
 =w_\gamma^2\bigl(J_{\gamma,\FO}-J_{\gamma,\SW}\bigr).
\end{equation}

\subsection{Optimal noise probes}
The four Pauli branch probabilities are
\begin{equation}
 p_{00}=(1-\gamma)^2,\qquad p_{10}=p_{01}=w_\gamma,
 \qquad p_{11}=\gamma^2.
\end{equation}
For any probe independent of $\gamma$, attach an orthogonal flag to each branch:
$\sum_h p_h(\gamma)\ket h\bra h\otimes U_h\rho U_h^\dagger$, where the
branch unitaries $U_h$ are independent of $\gamma$ in both architectures.
Only the branch probabilities depend on $\gamma$. The flagged QFI is therefore the Fisher information of two independent Bernoulli trials, $\sum_h(p_h')^2/p_h=2/w_\gamma$. Tracing out the flag proves $J_\gamma\leq2/w_\gamma$, including probes correlated with a reference.

For FO, a Bell input turns the four error branches into four orthogonal
Bell states, so a Bell measurement attains the bound. For SW, the product
input $\ket{+0}$ gives orthogonal outputs $\ket{+0},\ket{+1},\ket{-0}$
with probabilities $(1-\gamma)^2,2w_\gamma,\gamma^2$. Their classical
Fisher information is $2/w_\gamma$. The two branches with one flip have the same score $\partial_\gamma\log p_h$, so merging them preserves their Fisher information. This proves
Eq.~\eqref{eq:qfi-noise-optimal}.

\section{Fidelity and parameter information orderings}
\label{app:qfi-separation}

\subsection{Fidelity gain with a phase information loss}
Consider two identical amplitude damping channels with Kraus operators
\begin{equation}
 K_0=L_0=\begin{pmatrix}1&0\\0&1/2\end{pmatrix},\qquad
 K_1=L_1=\frac{\sqrt3}{2}\ket0\bra1.
\end{equation}
Choose $\ket m=\ket+$, $\ket n=-\ket-$ and the aligned family
\begin{equation}
 \ket{\chi_\varphi}=
 \frac{3\ket0\ket m+e^{i\varphi}\ket1\ket n}{\sqrt{10}}.
\end{equation}
Its concurrence is $3/5$. The phase is encoded before the noise by the
target generator $\ket n\bra n$. At $\varphi=0$, the input vector is
$(3,3,-1,1)^{\mathsf T}/\sqrt{20}$, and the outputs are
\begin{align}
 \rho_{\FO}&=\frac1{320}
 \begin{pmatrix}279&36&-3&12\\36&9&-12&3\\-3&-12&31&-4\\12&3&-4&1\end{pmatrix},
 \label{eq:ad-fo}\\
 \rho_{\SW}&=\frac1{320}
 \begin{pmatrix}279&36&-12&12\\36&9&-12&3\\-12&-12&31&-4\\12&3&-4&1\end{pmatrix}.
 \label{eq:ad-sw}
\end{align}
On this input family, the target phase encoding is equivalent to a control phase generated by $G_C=\ket1\bra1_C\otimes I_T$. Both joint channels commute with that control phase, so $\partial_\varphi\rho_X=i[G_C,\rho_X]$. At $\varphi=0$, the displayed states are real. Writing the SLD as $L_X=iA_X$ reduces its equation to $A_X\rho_X+\rho_X A_X=2[G_C,\rho_X]$, with $A_X$ real and antisymmetric. Solving this linear system and using $J_\varphi=\tr(\rho_XL_X^2)$ gives
\begin{align}
 P_{\FO}&=\frac{173}{320},&
 P_{\SW}&=\frac{1757}{3200},&
 \Delta P&=\frac{27}{3200}>0,\\
 J_{\varphi,\FO}&=\frac{3681}{67600},&
 J_{\varphi,\SW}&=\frac{4516389}{85181600}.&&
\end{align}
The phase information difference is
\begin{equation}
 J_{\varphi,\SW}-J_{\varphi,\FO}
 =-\frac{20613933}{14395690400}<0.
\end{equation}
The two nonzero commutators have squared diagonal expectation $3/64$
each in $m$. With $q=9/100$, their sum gives the positive fidelity gap
through Theorem~\ref{thm:qubit}. The negative QFI difference separates input fidelity from local phase distinguishability. The ancillary computational files supply the exact rational SLD matrices for Eqs.~\eqref{eq:ad-fo} and~\eqref{eq:ad-sw}.

\subsection{Product probes and parameter location}
For a product input $\tau_C\otimes\rho_T(\lambda)$ with
parameter-independent $\tau_C$, suppose the channels commute at each
parameter value. They may also depend on $\lambda$. The common target
marginal is $\sigma_T(\lambda)=(\Phi\Psi)(\rho_T(\lambda))$, and the
fixed-order output is $\tau_C\otimes\sigma_T(\lambda)$.
QFI monotonicity under partial trace gives
\begin{equation}
 J_{\lambda,\SW}\geq J_\lambda(\sigma_T)
 =J_{\lambda,\FO}.
\end{equation}
Applying the argument in each parameter direction gives the QFI matrix ordering. Together with Eq.~\eqref{eq:all-product-bound}, this gives opposite inequalities for input fidelity and QFI: the switched joint output can carry more parameter information while having lower fidelity with the input.

If the parameter is instead carried entirely by $\tau_C(\lambda)$,
with fixed target input and fixed channels, FO preserves the input QFI.
Data processing from the input gives
$J_{\lambda,\SW}\leq J_{\lambda,\FO}=J_\lambda(\tau_C)$.
The location of the encoded parameter therefore determines the direction of the QFI comparison.

\end{appendices}

\end{document}